\documentclass[12pt]{article}

\usepackage{color}
\usepackage{amsmath,enumerate}
\usepackage{amssymb}
\usepackage{epsfig}
\usepackage{subfigure}
\usepackage{pgf}
\usepackage{pgfcore}
\usepackage{pgfbaseimage}
\usepackage{pgfbaselayers}
\usepackage{pgfbasepatterns}
\usepackage{pgfbaseplot}
\usepackage{pgfbaseshapes}
\usepackage{pgfbasesnakes}
\usepackage{tikz}

\usepackage{a4wide}

\usepackage{algorithmic}
\usepackage{algorithm}
\def\paren#1{\left( #1 \right)}
\def\acc#1{\left\{ #1 \right\}}

\author{Louis Esperet\thanks{Laboratoire G-SCOP (Grenoble-INP, CNRS),
    Grenoble, France.}, Sylvain Gravier\thanks{Institut Fourier
    (Universit\'e
    Joseph Fourier, CNRS), St Martin d'H\`eres, France.}, Micka\"el
  Montassier\thanks{LaBRI (Universit\'e de Bordeaux, CNRS),
    Talence, France.},\\ Pascal Ochem\thanks{LRI (Universit\'e Paris Sud,
    CNRS), Orsay, France.}, Aline Parreau{$^\ddagger$\!\!}}

\title{Locally identifying coloring of graphs\thanks{This research is
    supported by the ANR IDEA, under contract {ANR-08-EMER-007},
    2009-2011.}}  
\date{\today}

\newenvironment{proof}{\par \noindent \textsc{Proof.} }{\hfill$\Box$\medskip}

\newtheorem{theorem}{Theorem}

\newtheorem{corollary}[theorem]{Corollary}

\newtheorem{proposition}[theorem]{Proposition}
\newtheorem{lemma}[theorem]{Lemma}
\newtheorem{conjecture}[theorem]{Conjecture}

\newtheorem{observation}[theorem]{Observation}

\newtheorem{question}[theorem]{Question}

\begin{document}

\maketitle

\renewcommand{\abstractname}{Abstract}

\begin{abstract}
 We introduce the notion of locally identifying coloring of a graph. A
 proper vertex-coloring $c$ of a graph $G$ is said to be {\em locally
   identifying}, if for any adjacent vertices $u$ and $v$ with
 distinct closed neighborhood, the sets of colors that appear in the
 closed neighborhood of $u$ and $v$ are distinct. Let
 $\chi_{\text{lid}}(G)$ be the minimum number of colors used in a
 locally identifying vertex-coloring of $G$. In this paper, we
 give several bounds on $\chi_{\text{lid}}$ for different families of
 graphs (planar graphs, some subclasses of perfect graphs, graphs with
 bounded maximum degree) and prove that deciding whether
 $\chi_{\text{lid}}(G)=3$ for a subcubic bipartite graph $G$ with
 large girth is an NP-complete problem.
\end{abstract}

\section{Introduction}

In this paper we focus on colorings allowing to distinguish the
vertices of a graph. In \cite{HS95}, Hor\v{n}\'ak and Sot\'ak
considered edge-coloring of a graph such that $(i)$ the edge-coloring
is proper (i.e. no adjacent edges receive the same color) and $(ii)$
for any vertices $u, v$ (with $u\ne v$) the set of colors assigned to
the edges incident to $u$ differs from the set of colors assigned to
the edges incident to $v$. Such a coloring is called a {\em
  vertex-distinguishing proper edge-coloring}. The minimum number of
colors required in any vertex-distinguishing proper edge-coloring of
$G$ is called the {\em observability} of $G$ and was studied for
different families of graphs
\cite{BRS03,BS97,CHS96,DTS02,FLS96,HS95,HS97}. This notion was then
extended to {\em adjacent vertex-distinguishing edge-coloring} where
Property $(ii)$ must be true only for pairs of adjacent vertices; see
\cite{ABN06,Hat05,ZLW02}.

In the present paper we introduce the notion of {\em locally
  identifying colorings}: a vertex-coloring is said to be {\em locally
  identifying} if $(i)$ the vertex-coloring is proper (i.e. no
adjacent vertices receive the same color), and $(ii)$ for any pair of
adjacent vertices $u,v$ the set of colors assigned to the closed
neighborhood of $u$ differs from the set of colors assigned to the
closed neighborhood of $v$ whenever these neighborhoods are
distinct. The {\em locally identifying chromatic number} of the graph
$G$ (or lid-chromatic number, for short), denoted by
$\chi_{\text{lid}}(G)$, is the smallest number of colors required in
any locally identifying coloring of $G$. In the following we study the
parameter $\chi_{\text{lid}}$ for different families of graphs, such
as bipartite graphs, $k$-trees, interval graphs, split graphs,
cographs, graphs with bounded maximum degree, planar graphs with high
girth, and outerplanar graphs.

\medskip
Let $G=(V,E)$ be a graph. For any vertex $u$, we denote by $N(u)$ its
neighborhood and by $N[u]$ its \emph{closed neighborhood} ($u$
together with its adjacent vertices) and by $d(u)$ its degree.  Let
$c$ be a vertex-coloring of $G$. For any $S\subseteq V$, let $c(S)$ be
the set of colors that appear on the vertices of $S$. More formally, a
locally identifying coloring of $G$ (or a lid-coloring, for short) is
proper vertex-coloring $c$ of $G$ such that for any edge $uv$,
$N[u]\ne N[v] \Rightarrow c(N[u])\ne c(N[v])$.  Observe that the
lid-chromatic number of a graph $G$ is the maximum of the
lid-chromatic numbers of its connected components. Hence, in the
proofs of most of our results it will be enough to restrict ourselves
to connected graphs. A graph $G$ is {\em $k$-lid-colorable} if it
admits a locally identifying coloring using at most $k$ colors. Notice
the following:

\begin{observation}\label{obs:2lid}
A connected graph $G$ is 2-lid-colorable if and
only if $G$ has at most two vertices.
\end{observation}

\begin{proof}
Let $G$ be a connected graph with a $2$-lid-coloring $c$ and at least
3 vertices. Consider an edge $uv$. Then we have $N[u]\neq N[v]$, since
otherwise $G$ would contain a triangle and then we would have
$\chi_{\text{lid}}(G)\geq \chi(G) \geq 3$. Since $c$ is a 2-coloring
and $N[u]$ and $N[v]$ both contain $u$ and $v$, we have
$c(N[u])=c(N[v])=\{c(u),c(v)\}$, a contradiction.

The other implication is trivial.
\end{proof}

\medskip
Note that locally identifying coloring is not hereditary. For
instance, if $P_n$ denotes the path on $n$ vertices, then
$\chi_{\text{lid}}(P_5)=3$ whereas $\chi_{\text{lid}}(P_4)=4$.

\medskip
In Section \ref{sec:bip}, we prove that every bipartite graph has
lid-chromatic number at most 4.  Moreover, deciding whether a
bipartite graph is 3-lid-colorable is an NP-complete problem, whereas
it can be decided in linear time whether a tree is
3-lid-colorable.

\medskip 
In general, $\chi_{\text{lid}}$ is not bounded by a function of the
usual chromatic number $\chi$. Nevertheless it turns out that for
several nice classes of graphs such a function exists: we study
$k$-trees (Section \ref{sec:ktree}), interval graphs (Section
\ref{sec:inter}), split graphs (Section \ref{sec:split}), cographs
(Section \ref{sec:cographs}), and give tight bounds in each of these
cases. We also conjecture that every chordal graph $G$ has a
lid-coloring with $2 \chi(G)$ colors.

\medskip 
Section \ref{sec:bmaxd} is dedicated to graphs with bounded
maximum degree. We prove that the lid-chromatic number of graphs with
maximum degree $\Delta$ is $O(\Delta^3)$ and that there are examples
with lid-chromatic number $\Omega(\Delta^2)$.

\medskip
In Section \ref{sec:planar}, we study graphs with a
topological structure. Our result on 2-trees does not give any
information on outerplanar graphs, since lid-coloring is not monotone
under taking subgraphs. So we use a completely different strategy to
prove that outerplanar graphs and planar graphs with large girth have
lid-colorings using a constant number of colors.

\medskip
Finally, in Section \ref{conn}, we propose a tool allowing to
extend the lid-colorings of the 2-connected components of a graph to
the whole graph.

\section{Bipartite graphs}\label{sec:bip}

This section is dedicated to bipartite graphs. The main
interest of the study of bipartite graphs here comes from the
following lemmas:

\begin{lemma}\label{lem:bip1}
If a connected graph $G$ satisfies $\chi_{\text{lid}}(G)\le 3$, then $G$ is
either a triangle or a bipartite graph.
\end{lemma}

\begin{proof}
Consider a 3-lid-coloring $c$ of $G$ with colors $1,2,3$. By
Observation~\ref{obs:2lid}, we can assume that $G$ has at least three
vertices.

Define the coloring $c'$ by $c'(x)=|c(N[x])|$ for any vertex
$x$. Since $G$ is connected, $c'(x) \in \{2,3\}$ for any vertex $x$. If
two adjacent vertices $u,v$ satisfy $c'(u)=c'(v)=3$, then
$c(N[u])=c(N[v])=\{1,2,3\}$, and if $c'(u)=c'(v)=2$, then
$c(N[u])=c(N[v])=\{c(u),c(v)\}$. It follows that $c'$ is a proper
2-coloring of $G$ unless $N[u]=N[v]$ for some edge $uv$. In this case,
since $G$ does not consist of the single edge $uv$, there exists a
vertex $w$ adjacent to $u$ and $v$. But then
$c(N[u])=c(N[v])=c(N[w])=\{1,2,3\}$, which implies that
$N[u]=N[v]=N[w]$. This is only possible if $G$ is
a triangle.
\end{proof}

Indeed, more can be said about the color classes in a 3-lid-coloring
of a (bipartite) graph:

\begin{lemma}\label{lem:bip2}
Let $G$ be a 3-lid-colorable connected bipartite graph on at least
three vertices, with bipartition $\{U,V\}$, and let $c$ be a
3-lid-coloring of $G$ with colors $1,2,3$. Then $G$ has a vertex $u$
with $c(N[u])=\{1,2,3\}$ and if $u\in U$, then $c(U)=\{c(u)\}$ and
$c(V)=\{1,2,3\} \setminus \{c(u)\}$.
\end{lemma}

\begin{proof}
Let $uv$ be an edge of $G$. We have $N[u]\neq N[v]$ because $G$ is a
bipartite connected graph on at least three vertices. Then
$c(N[u])=\{1,2,3\}$ or $c(N[v])=\{1,2,3\}$. Without loss of
generality, assume that $c(N[u])=\{1,2,3\}$ and $c(u)=1$.  Then all
the neighbors of $u$ must be colored 2 or 3, and the vertices at
distance two from $u$ must be colored 1 (otherwise there would be a
neighbor $w$ of $u$ with $c(N[w])=\{1,2,3\}$ and $N[u]\neq N[w]$).
Iterating this observation, we remark that all the vertices at even
distance from $u$ must be colored 1, while the vertices at odd
distance from $u$ must be colored either 2 or 3, which yields the
conclusion.
\end{proof}

As a corollary we obtain a precise description of
3-lid-colorable trees.

\begin{corollary}\label{tree}
A tree $T$ with at least 3 vertices is 3-lid-colorable if and only if
the distance between every two leaves is even.
\end{corollary}

\begin{proof}
Observe that for each leaf $u$ of $T$,
we have $|c(N[u])|=2$ in any proper coloring $c$ of $T$, so by
Lemma~\ref{lem:bip2} the distance between every two leaves is even.

Now assume that the distance between every two leaves of $T$ is
even, and fix a leaf $u$ of $T$. Let $c$ be the 3-coloring of $T$
defined by $c(v)=2$ if $d(u,v)$ is odd, $c(v)=1$ if $d(u,v) \equiv 0
\bmod 4$, and $c(v)=3$ if $d(u,v)\equiv 2 \bmod 4$. The coloring $c$ is
clearly proper, and we have $c(N[v])=\{1,2\}$ if $d(u,v) \equiv 0 \bmod
4$, and $c(N[v])=\{2,3\}$ if $d(u,v) \equiv 2 \bmod 4$. If $v$ is a
vertex at odd distance from $u$, then $v$ is not a leaf and
$c(N[v])=\{1,2,3\}$. As a consequence, $c$ is a 3-lid-coloring of $T$.
\end{proof}

Another class of bipartite graphs that behaves nicely with regards to
locally identifying coloring is the class of graphs obtained by taking
the Cartesian product of two bipartite graphs. For two graphs
$G_1=(V_1,E_1)$ and $G_2=(V_2,E_2)$, the Cartesian product of $G_1$
and $G_2$, denoted by $G_1 \square G_2$, is the graph with vertex set
$V_1 \times V_2$, in which two vertices $(u_1,u_2)$ and $(v_1,v_2)$
are adjacent whenever $u_2=v_2$ and $u_1v_1 \in E_1$, or $u_1=v_1$ and
$u_2v_2 \in E_2$. 

\begin{theorem}\label{prod}
If $G_1$ and $G_2$ are bipartite graphs without isolated vertices,
then $G_1\square G_2$ is 3-lid-colorable.
\end{theorem}

\begin{proof}
Let $\{U_1,V_1\}$ and $\{U_2,V_2\}$ be the partite sets of $G_1$ and
$G_2$, respectively. Then $G_1\square G_2$ is a bipartite graph with
partition $\{(U_1\times U_2) \cup (V_1\times V_2),(U_1\times V_2) \cup
(V_1\times U_2)\}$ and because there are no isolated vertices in $G_1$
and $G_2$, each vertex of $(U_1\times U_2) \cup (V_1\times V_2)$ has a
neighbor in $U_1\times V_2$ and a neighbor in $V_1\times U_2$.

We define $c$ by $c(u)=1$ if $u\in (U_1\times U_2) \cup (V_1\times
V_2)$, $c(u)=2$ if $u \in U_1\times V_2$, and $c(u)=3$ if $u \in
V_1\times U_2$.  Then $c$ is a lid-coloring of $G_1\square G_2$:
$c(N[u])=\{1,2,3\}$ for vertices of $(U_1\times U_2) \cup (V_1\times
V_2)$, $c(N[u])=\{1,2\}$ for vertices of $U_1\times V_2$ and
$c(N[u])=\{1,3\}$ for vertices of $V_1\times U_2$.

By Observation~\ref{obs:2lid}, $G_1\square G_2$ does not have a
$2$-lid-coloring.
\end{proof}

As a corollary, we obtain that hypercubes and grids in any dimension
are 3-lid-colorable. We now focus on bipartite graphs that are not
3-lid-colorable.

\begin{theorem}\label{th-bip}
If $G$ is a bipartite graph, then $\chi_{\text{lid}} (G)\leq 4$.
\end{theorem}

\begin{proof}
We can assume that $G$ is a connected graph with at least five
vertices. Then there exists a vertex $u$ of $G$ that is not adjacent
to a vertex of degree one. For any vertex $v$ of $G$, set $c(v)$ to be
the element of $\{0,1,2,3\}$ congruent with $d(u,v)$ modulo 4. We
claim that $c$ is a lid-coloring of $G$.  Since $G$ is bipartite, $c$
is clearly a proper coloring. Let $v,w$ be two adjacent vertices in
$G$. We may assume that they are at distance $k\ge 0$ and $k+1$ from
$u$, respectively. If $k=0$, then $v=u$ and $w$ has a neighbor at
distance two from $u$, so $c(N[v])=\{0,1\}$ and
$c(N[w])=\{0,1,2\}$. If $k\ge 1$, then $(k-1) \bmod 4$ is in $c(N[v])$
but not in $c(N[w])$, so $c(N[v])\ne c(N[w])$.
\end{proof}

We now prove that deciding whether a bipartite graph is $3$ or
$4$-lid-colorable is a hard problem.

\begin{theorem}\label{bip-np}
For any fixed integer $g$, deciding whether a bipartite graph with
girth at least $g$ and maximum degree 3 is 3-lid-colorable is an
NP-complete problem.
\end{theorem}

\begin{proof}
We recall that a 2-coloring of a hypergraph
$\mathcal{H}=(\mathcal{V},\mathcal{E})$ is a partition of its vertex
set $\mathcal{V}$ into two color classes such that no edge in
$\mathcal{E}$ is monochromatic. We reduce our problem to the
NP-complete problem of deciding the 2-colorability of $3$-uniform
hypergraphs~\cite{Lov73}.

Let $\mathcal{H}=(\mathcal{V},\mathcal{E})$ be a hypergraph with at
least one hyperedge.  We construct the bipartite graph $G=(V,E)$ in
the following way. To each vertex $v\in \mathcal{V}$, we associate a
path $P_v$ with vertices $\{v_0,\dots,v_{4t}\}$ in $G$ (where $t$ will
depend on the degree of $v$ in $\mathcal{H}$ and the girth $g$ we want
for $G$). All the paths $P_v$ are built on disjoint sets.  To each
hyperedge $e\in \mathcal{E}$, we associate a vertex $w_e$ in $G$. If a
hyperedge $e$ contains a vertex $v$ in $\mathcal{H}$, then we add an
edge in $G$ between $w_e$ and a vertex $v_i$ of $P_v$ for some index
$i\equiv 2\bmod 4$.  We require that a vertex $v_i$ on a path $P_v$ is
adjacent to at most one vertex corresponding to a hyperedge containing
$v$. It follows that the graph $G$ is bipartite with maximum
degree 3.  Moreover, we can construct $G$ in polynomial time and
ensure that the girth of $G$ is at least $g$ by leaving enough space
(at least $g/2$ vertices of degree two) between any two consecutive
vertices of degree 3 on the paths $P_v$.

We shall prove that $\mathcal{H}$ is 2-colorable if and only if
$\chi_{\text{lid}}(G)=3$.

Assume first that $\mathcal{H}$ admits a 2-coloring $\mathcal{C}: V
\rightarrow \{1,2\}$. We define the following 3-coloring $c$ of $G$
such that $c(v_{i\equiv 2 \bmod 4})=\mathcal{C}(v)$, $c(v_{i \equiv 0
  \bmod 4})=3-\mathcal{C}(v)$, $c(v_{i\equiv 1 \bmod 2})=3$ if $v\in
V$, and $c(w_e)=3$ for all vertices $w_e$ with $e\in \mathcal{E}$.
Let us check that $c$ is a lid-coloring of $G$.  We have
$c(N[w_e])=\acc{1,2,3}$ since $c(w_e)=3$ and $w_e$ is adjacent to a
vertex colored 1 and to a vertex colored 2 because of the 2-coloring
of $\mathcal{H}$.  Also, $c(N[v_{i\equiv 1 \bmod 2}])=\acc{1,2,3}$,
$c(N[v_{i\equiv 2 \bmod 4}])=\acc{\mathcal{C}(v),3}$, and $c(N[v_{i
    \equiv 0 \bmod 4}])=\acc{3-\mathcal{C}(v),3}$.  So, for every edge
$uv$ in $G$, we have $c(N[u])\ne c(N[v])$.

Conversely, assume that $G$ (with bipartition $\{U,V\}$) admits a
lid-coloring $c$ using colors $1,2,3$.  By Lemma~\ref{lem:bip2}, we
can assume that $c(U)=\{1,2\}$ and $c(V)=\{3\}$, and that the vertices
of degree one in $G$ are in $U$. This implies that $c(v_{i\equiv 2
  \bmod 4})\in\acc{1,2}$, $c(v_{i\equiv 0 \bmod 4})=3-c(v_{i \equiv 2
  \bmod 4})$, and $c(v_{i \equiv 1 \bmod 2})=c(w_e)=3$.  Hence, the
vertex-coloring of $\mathcal{V}$, in which each vertex $v$ receives
the color $c(v_{i\equiv 2 \bmod 4})$, is 2-coloring of the
hypergraph~$\mathcal{H}$.
\end{proof}

It turns out that the connection between 3-lid-coloring and hypergraph
2-coloring highlighted in the proof of Theorem~\ref{bip-np} has
further consequences. For a connected bipartite graph $G$ with
bipartition $\{U,V\}$, let $\mathcal{H}_U$ be the hypergraph with vertex
set $U$ and hyperedge set $\{N(v), v\in V\}$. A direct consequence of
Lemmas~\ref{lem:bip1} and~\ref{lem:bip2} is that a connected graph $G$
distinct from a triangle is 3-lid-colorable if and only if it is
bipartite (say with bipartition $\{U,V\}$) and at least one of
$\mathcal{H}_U$ and $\mathcal{H}_V$ is 2-colorable.

A consequence of a result of Moret~\cite{Mor88} (see also~\cite{ABB08}
for further details) is that if $G$ is a subcubic bipartite planar
graph with bipartition $\{U,V\}$, then we can check in polynomial time
whether $\mathcal{H}_U$ (or $\mathcal{H}_V$) is 2-colorable. As a
counterpart of Theorem~\ref{bip-np}, this implies:

\begin{theorem}\label{th-planar-bip}
It can be checked in polynomial time whether a planar graph $G$ with
maximum degree three is 3-lid-colorable.
\end{theorem}

It was proved by Burstein~\cite{Bur74} and Penaud~\cite{Pen75} that
every planar hypergraph in which all hyperedges have size at least
three is 2-colorable, and Thomassen~\cite{Tho92} proved that for any
$k \ge 4$ any $k$-regular $k$-uniform hypergraph is 2-colorable. As a
consequence, we obtain the following two results:

\begin{theorem}\label{th-planar-deg}
Let $G$ be a bipartite planar graph with bipartition $\{U,V\}$ such that
all vertices in $U$ or all vertices in $V$ have degree at least three.
Then $G$ is 3-lid-colorable.
\end{theorem}

\begin{theorem}\label{th-reg}
For $k\ge 4$, a $k$-regular graph is 3-lid-colorable if and only if it
is bipartite.
\end{theorem}

Since bipartite graphs have bounded lid-chromatic number, a natural
question is whether $\chi_{\text{lid}}$ is upper-bounded by a function of the
(usual) chromatic number. However, this is not true, since the graph
$G$ obtained from a clique on $n$ vertices by subdividing each edge
exactly twice has $\chi_{\text{lid}}(G)=n$ (it suffices to observe that two
vertices of the initial clique cannot have the same color in the
subdivided graph), whereas it is 3-colorable. This example also shows
that if the edges of a graph $G$ are partitioned into two sets $E_1$
and $E_2$, and the subgraphs of $G$ induced by $E_1$ and $E_2$ have
bounded lid-chromatic number, then $\chi_{\text{lid}}(G)$ is not necessarily
bounded.

We propose the following conjecture relating $\chi_{\text{lid}}$ and $\chi$
for highly structured graphs. A graph is \emph{chordal} if it does not
contain an induced cycle of length at least four.

\begin{conjecture}\label{conj:chordal}
  For any chordal graph $G$, $\chi_{\text{lid}}(G)\leq 2\chi(G)$.
\end{conjecture}

The next three sections are dedicated to important subclasses of
chordal graphs for which we are able to verify
Conjecture~\ref{conj:chordal}.

\section{$k$-trees \label{sec:ktree}}

This section is devoted to the study of $k$-trees.  A {\em $k$-tree}
is a graph whose vertices can be ordered $v_1,v_2,\ldots,v_n$ in
such a way that the vertices $v_1$ up to $v_{k+1}$ induce a $(k+1)$-clique
and for each $k+2 \le i \le n$, the neighbors of $v_i$ in
$\{v_j\,|\,j<i\}$ induce a $k$-clique. By definition, for every $k+1 \leq
i \le n$ the graph $G_i$ induced by $\{v_j\,|\,j\le i\}$ is a $k$-tree and
every $k$-clique in a $k$-tree is contained in a $(k+1)$-clique.

\begin{theorem}\label{th-ktree}
If $G$ is a $k$-tree, then $\chi_{\text{lid}}(G)\leq 2k+2$.
\end{theorem}

\begin{proof}
  In this proof the colors are the integers modulo $2k+2$. In
  particular, this implies that the function on integers $x \mapsto
  x+k+1$ is an involution.

Let $v_1,\ldots,v_n$ be the $n$ vertices of $G$ ordered as above.

We construct the following coloring $c$ of $G$
iteratively for $1 \le i \le n$. If $i \le k+1$, then we set
$c(v_i)=i$. Suppose $i \ge k+2$. Let $C$ be the neighborhood of $v_i$
in $G_{i}$. Since $G_{i-1}$ is a $k$-tree, the clique $C$ is contained
in a $(k+1)$-clique $C'$ of $G_{i-1}$. Let $\{v_j\}=C' \setminus
C$. We set $c(v_i)=c(v_j)+k+1$ (we may have several choices for $C'$
and thus for $j$).

We now prove that $c$ is a lid-coloring of $G$. Throughout the
procedure, the following two properties remain trivially satisfied:
(i) $c$ is a proper coloring of $G$, and (ii) no vertex colored $i$
has a neighbor colored $i+k+1$. Consider an edge $v_iv_j$ of $G$ with
$N[v_i] \ne N[v_j]$. We may assume without loss of generality that
some neighbors of $v_i$ are not adjacent to $v_j$. If $i,j \le k+1$,
then consider the minimum index $\ell$ such that $v_\ell$ is a
neighbor of $v_i$ not adjacent to $v_j$. By definition of $c$
and minimality of $\ell$, we have
$c(v_j)=c(v_\ell) +k+1$. Otherwise we can assume that $j>i$ and
$j>k+1$. Let $C$ be the neighborhood of $v_j$ in $G_j$. By definition
of $c$, there exists a $(k+1)$-clique $C'$ of $G_{j-1}$ containing $C$
such that $c(v_j)=c(v_\ell) +k+1$, where $C' \setminus
C=\{v_\ell\}$. In both cases, $c(v_\ell) \in c(N[v_i])$ while
$c(v_\ell) \not\in c(N[v_j])$ by Property~(ii). Hence, $c$ is a
lid-coloring of $G$.
\end{proof}

 \begin{figure}[htbp]
    \begin{center}
      \subfigure[\label{fig:inter}]{\includegraphics[scale=0.8]{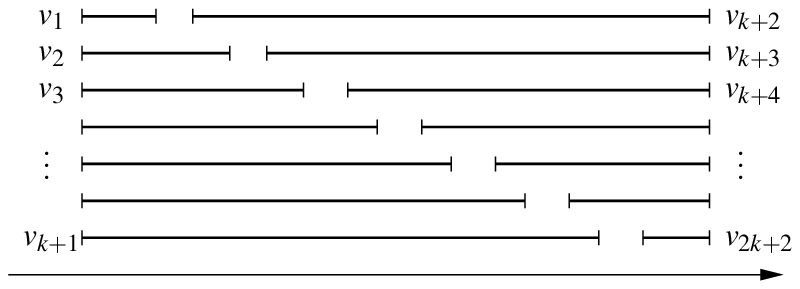}}\hspace{1cm}
      \subfigure[\label{fig:perm}]{\includegraphics[scale=0.8]{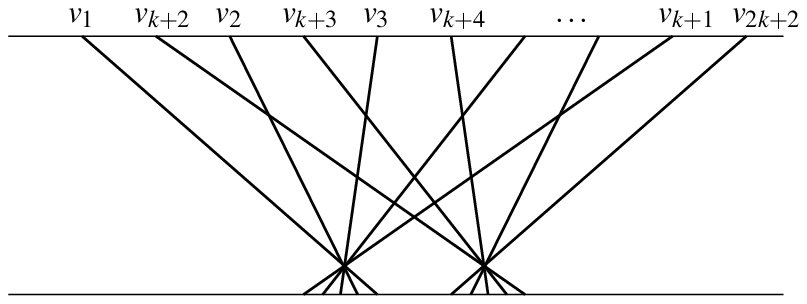}}
      \caption{The graph $P_{2k+2}^k$ as an interval graph (a) and
        as a permutation graph (b).}
    \end{center}
  \end{figure}

For fixed $t$, the fact that a graph admits a lid-coloring with at
most $t$ colors can be easily expressed in monadic second-order
logic. Thus Theorem~\ref{th-ktree} together with~\cite{CMR00} imply
that for fixed $k$, the lid-chromatic number of a $k$-tree can be
computed in linear time. Another remark is that for trees,
Theorem~\ref{th-ktree} provides the same 4-lid-coloring as
Theorem~\ref{th-bip}.

\medskip

For any two integers $k,\ell \ge 1$, we define $P_{\ell}^{k}$ as the
graph with vertex set $v_1,\ldots,v_{\ell}$ in which $v_i$ and $v_j$
are adjacent whenever $|i-j|\le k$. The graph $P_{2k+2}^{k}$ is
clearly a $k$-tree: it can be constructed from the clique formed by
$v_1,\ldots,v_{k+1}$ by adding at each step $k+2 \le i \le 2k+2$ a
vertex $v_i$ adjacent to $v_{i-k},\ldots,v_{i-1}$. The graph
$P_{2k+2}^{k}$ is also an interval graph (see Figure~\ref{fig:inter})
and a permutation graph (see Figure~\ref{fig:perm}). We now prove that
the graph $P_{2k+2}^{k}$ also provides an example showing that
Theorem~\ref{th-ktree} is best possible.

\begin{proposition}
For any $k \ge 1$, we have $\chi_{\text{lid}}(P_{2k+2}^{k})=2k+2$.
\end{proposition}

\begin{proof}
Let $c$ be a lid-coloring of $P_{2k+2}^{k}$. Without loss of
generality we have $c(v_i)=i$ for each $1 \le i \le k+1$. Observe that
for any $1 \le i \le k$, the symmetric difference between $N[v_{i}]$
and $N[v_{i+1}]$ is precisely $\{v_{i+k+1}\}$.  In addition,
$N[v_{i}]=\{v_1,\dots,v_{i+k}\}$ and so $c(N[v_i])$ contains colors
$1$ up to $k+1$.  Therefore, $c(v_i) > k+1$ whenever $k+2\le i \le
2k+1$. And we can assume that $c(v_i)=i$ for any $1 \le i \le 2k+1$.

Let $\alpha=c(v_{2k+2})$, and assume for the sake of contradiction
that $\alpha \ne 2k+2$.  Since vertices $v_{k+2},\ldots,v_{2k+2}$
induce a clique, we have $\alpha \le k+1$. The symmetric difference
between $N[v_{\alpha+k}]$ and $N[v_{\alpha+k+1}]$ is precisely
$\{v_\alpha\}$ if $\alpha\ge 2$ and is $\{v_1,v_{2k+2}\}$ if
$\alpha=1$.  In both cases, $c(v_{2k+2})=c(v_\alpha)=\alpha$ would
imply that $c(N[v_{\alpha+k}])=c(N[v_{\alpha+k+1}])$, a contradiction.
\end{proof}

\section{Interval graphs}\label{sec:inter}

In this section, we prove that the previous example is also extremal
for the class of interval graphs.

\begin{theorem}
  For any interval graph $G$, $\chi_{\text{lid}}(G)\leq 2\omega(G)$.
\end{theorem}

\begin{proof} 
Let $k=\omega(G)$. In this proof the colors are the integers modulo
$2k$.  Let $G$ be a connected interval graph on $n$ vertices.  We
identify the vertices $v_1,\ldots,v_n$ of $G$ with a family of
intervals $(I_i=[a_i,b_i])_{1 \le i \le n}$ such that $v_iv_j$ is an
edge of $G$ precisely if $I_i$ and $I_j$ intersect. We may assume that
$a_1 \le a_2 \le \ldots \le a _n$.  Without loss of generality, we can
assume that if $a_{i}<a_{j}$ and $I_i \cap I_j \ne \varnothing$, then
there exists an interval $I_{\ell}$ such that $a_{i}\leq
b_{\ell}<a_{j}$; otherwise, we can change $I_{j}$ to the interval
$[a_{i},b_{j}]$ and the intersection graph remains the same. By a
similar argument, we can also assume that if $b_{j}<b_{i}$ and $I_i
\cap I_j \ne \varnothing$, then there exists an interval $I_{\ell}$
such that $b_{j}<a_{\ell}\leq b_{i}$.

\medskip

Let $\{a_1=a_{t_1} < a_{t_2} < \ldots < a_{t_s} \} $ be the set of
left ends.  At each step $i=1,\ldots,s$, we color all the intervals
starting at $a_{t_i}$. We first color the intervals starting at
$a_{t_1}$ with distinct colors in $\{0,\ldots,k-1\}$. Assume we have
colored all the intervals starting before $a_{t_i}$. Now, we
color all the intervals $\mathcal{I}(t_i)$ starting at
$a_{t_i}$. First, we define the following subsets of intervals:
\begin{itemize} 

\item $\mathcal V(t_{i})$: intervals $I_{j}$ such that
  $a_{j}<a_{t_{i}}\leq b_{j}$,

\item $\mathcal U(t_{i})$: intervals $I_{j}$ such that
  $a_{t_{i-1}}\leq b_{j} < a_{t_i}$,

\item $\mathcal T(t_{i})$: intervals $I_{j}$ of $\mathcal U(t_{i})$
  such that there is an interval $I_{\ell}$ in $\mathcal V(t_{i})$ with
  $a_{j}=a_{\ell}$.

\end{itemize} 

Note that $\mathcal V(t_{i})$ is the set of intervals that are already
colored and intersect $\mathcal{I}(t_i)$.  The set $\mathcal U(t_{i})$
is a subset of intervals already colored that intersect all the
intervals of $\mathcal V(T_{i})$. It is not empty (take any interval
finishing before $a_{t_{i}}$ with rightmost right end).  Necessarily,
all the intervals of $\mathcal U(t_{i})$ have the same right end
because no interval starts between $a_{t_{i-1}}$ and $a_{t_{i}}$.
Finally, if $\mathcal T(t_{i})\ne \varnothing$, then let $I_{0}$ be an
interval of $\mathcal T(t_{i})$ with leftmost left end, and otherwise
let $I_{0}$ be any interval of $\mathcal U(t_{i})$. Let $c_{0}$ be the
color of $I_{0}$. Note that any interval of $\mathcal U(t_{i})$ and
$\mathcal V(t_{i})$ intersects $I_{0}$, and thus has color $c_{0}$ in
its neighborhood.  We can now color the intervals of
$\mathcal{I}(t_{i})$.  We color with color $c_{0}+k$ one of the
intervals having the rightmost right end. We color the other intervals
with colors in $\{0,\ldots,2k-1\}$ such that no vertex with color $j$
is adjacent to a vertex with color $j$ or $j+k$ (this is always
possible since intervals of $\mathcal{V}({t_i}) \cup
\mathcal{I}(t_{i})$ induce a clique of size at most $k$).  This
coloring $c$ is clearly a proper $2k$-coloring and there is no vertex
with color $j$, $0\leq j \leq k-1$, adjacent to a vertex with color
$j+k$.

\medskip We now show that $c$ is a lid-coloring of $G$. Let $I_{i}$
and $I_{j}$ be two intersecting intervals with $N[I_{i}]\neq
N[I_{j}]$.  Assume first that $a_{i}\neq a_{j}$. Without loss of
generality, $a_{i}<a_{j}$. During the process, when $I_{j}$ is
colored, an interval $I_{\ell}$ also starting at $a_{j}$ is colored
with a color $c_{0}+k$ such that $c_{0}\in c(N[I_{i}])$. Necessarily,
$I_{j}\subseteq I_{\ell}$ since $I_{\ell}$ has the rightmost right end
among all intervals starting at $a_{j}$. So $c_{0}+k\in c(N[I_{j}])$
but $c_{0}\notin c(N[I_{\ell}])$ and so $c_{0}\notin
c(N[I_{j}])$. Hence, $c(N[I_{i}])\neq c(N[I_{j}])$.  Assume now that
$a_{i}=a_{j}$. Without loss of generality, $b_{j}<b_{i}$ and so $I_{j}
\subseteq I_{i}$.  Let $a_{t_{\ell}}$ be the leftmost left end such
that $b_{j}<a_{t_{\ell}}\leq b_{i}$ (it exists because $N[I_{i}]\neq
N[I_{j}]$).  Then we have $I_{i}\in \mathcal V(t_{\ell})$ and
$I_{j}\in \mathcal T(t_{\ell})$.  By construction, one of the
intervals of $\mathcal I(t_{\ell})$, say $I$, will receive the color
$c_{0}+k$ where $c_{0}$ is the color of an interval $I_{0}\in \mathcal
T(t_{\ell})$. Necessarily, $I_{j}\subseteq I_{0}$ and $c_{0}\in
c(N[I_{j}])\subseteq c(N[I_{i}])$.  We also have $c_{0}+k\in
c(N[I_{i}])$ because $I_{i}$ is a neighbor of $I$. But $c_{0}+k\notin
c(N[I_{j}])$ since $c_{0}+k\notin c(N[I_{0}])$ and $I_{j}\subseteq
I_{0}$. Hence, $c(N[I_{i}])\neq c(N[I_{j}])$.
\end{proof}

\section{Split graphs \label{sec:split}}

A {\em split graph} is a graph $G=(K\cup S,E)$ whose vertex set can be
partitioned into a clique $K$ and an independent set $S$.  In the
following, we will always consider partitions $K\cup S$ with $K$ of
maximum size.  A split graph is a chordal graph and its clique number and
chromatic number are equal to $\vert K \vert$.  We prove that it is lid-colorable
with $2\vert K\vert -1$ colors.

We say that a set $S'\subseteq S$ {\em discriminates} a set
$K'\subseteq K$ if for any $u,v \in K'$ with $N[u] \ne N[v]$, we also
have $N[u]\cap S' \ne N[v]\cap S'$. The following theorem is due to
Bondy:

\begin{theorem}[\cite{Bon72,CCHL08}]\label{bondy}
  If $A_1,A_2,\ldots,A_n$ is a family of $n$ distincts subsets of a set
  $\mathcal{A}$ with at least $n$ elements, then there is a subset
  $\mathcal{A}'$ of $\mathcal{A}$ of size $n-1$ such that all the sets
  $A_i\cap \mathcal{A}'$ are distinct.
\end{theorem}

\begin{corollary}
  Let $G=(K\cup S,E)$ be a split graph.  For any $K' \subseteq K$,
  there is a subset $S'$ of $S$ of size at most $|K'|-1$ such that $S'$
  discriminates $K'$.
\end{corollary}

\begin{proof}
  We apply Theorem~\ref{bondy} to the (at most) $|K'|$ pairwise
  distinct sets among \{$N[v]\cap S \, | \, v\in K'\}$.
\end{proof}

One can easily show that every split graph $G$ has lid-chromatic number
at most $2|K|$ by giving colors $1,\ldots,|K|$ to the vertices of $K$,
colors $|K|+1,\ldots,|K|+k-1$, for some $k\le |K|$, to the vertices of
a smallest discriminating set $S'\subseteq S$ of $K$, and finally color $|K|+k$
to the vertices of $S \setminus S'$. 

\medskip
We now prove the following stronger result:

\begin{theorem}\label{split}
Let $G=(K\cup S,E)$ be a split graph. If $\omega(G) \geq 3$ or if $G$
is a star, then $\chi_{\text{lid}}(G)\leq 2\omega(G)-1$.
\end{theorem}

\begin{proof}
Assume that $|K|=k$ and denote the vertices of $K$ by
$v_1,\ldots,v_k$. If $k=1$, then $G$ has no edges and it is clear that
$\chi_{\text{lid}}(G)\leq 1$. If $G=K_{1,n}$, then
$\chi_{\text{lid}}(G)\leq 3$ by Corollary~\ref{tree}. So we can assume
that $k\geq 3$.  If $|S|\leq k-1$ or if $S$ contains a set of size at
most $k-2$ that discriminates $K$, then the result is
trivial. Therefore, we assume that $|S|\geq k$ and consider a minimal
set $S_1$ that discriminates $K$. We can assume that the set $S_1$ has
size precisely $k-1$ and there is no edge $uv$ with
$N[u]=N[v]$. Indeed, if $N[u]=N[v]$ for an edge $uv$, then any set
discriminating $K\setminus\{v\}$ discriminates also $k$.  We consider
two cases.

\medskip
{\bf Case 1.} There is a vertex $x \in S\setminus S_1$ of degree $k-1$
and a neighbor $v_i\in K$ of $x$ such that $N[v_i]\cap
S_1=\varnothing$.  Without loss of generality, we can assume that
$v_i=v_{k-1}$ and that $K \setminus N(x)=\{v_k\}$. Let $S_x=\{y\in S,
N(y)=N(x)=K \setminus \{v_k\}\}$. We have $S_x\cap S_1= \varnothing$
(recall that $v_{k-1}$ has no neighbor in $S_1$) and by definition of
$S_1$, for each vertex $v_i\neq v_{k-1}$, $N[v_i]\cap S_1\neq
\varnothing$ ($S_1$ is a discriminating set).

Let $K_1=K\setminus \{v_{k-1},v_k\}$, and let $S_2$ be a subset of
$S_1$ of size at most $|K_1|-1=k-3$ that discriminates $K_1$. Let
$S'=S\setminus(S_1\cup S_x)$.  We define a coloring $c$ as follows:
\begin{itemize}
\item for $1 \le i \le k$, $c(v_i)=i$;
\item assign pairwise distinct colors from $k+1,\ldots,2k-3$ to the
  vertices of $S_2$;
\item for $u\in S_1\setminus S_2$, $c(u)=2k-2$;
\item for $u\in S_x$, $c(u)=2k-1$;
\item for $u\in S'$, take $v_i \in K \setminus N(u)$ ($v_i$ exists by
 maximality of $K$), and set $c(u)=c(v_i)$.
\end{itemize}

Then $c$ is a proper coloring of $G$. We show that $c$ is a
lid-coloring of $G$.  First observe that for each vertex $v_i$ of $K$,
$c(N[v_i])$ contains one color of $\{k+1,\ldots,2k-1\}$. Indeed $2k-1 \in
c(N[v_{k-1}])$ and if $v_i\neq v_{k-1}$, then $N[v_i]\cap S_1 \neq
\varnothing$ and therefore $c(N[v_i])\cap \{k+1,\ldots,2k-2\}\neq
\varnothing$.  This implies that for each $v_i \in K$, $c(N[v_i])$ is
distinct from all $c(N[y])$, $y\in S$. In fact, either $c(y)\in c(K)$
and then $c(N[y])\subseteq c(K)$, or $c(y)\notin c(K)$ but then there
is at least one color of $c(K)$ that $c(N[y])$ does not contain.
Furthermore, $c(N[v_{k}])$ is different from all the sets $c(N[v_{i}])$
with $i\neq k$ because $2k-1 \in c(N[v_{i}])$ and $2k-1 \notin
c(N[v_{k}])$. The set $c(N[v_{k-1}])$ is different from all the sets
$c(N[v_{i}])$ with $i\neq k-1$ because $c(N[v_{k-1}])$ contains no
color of $c(S_1)$ whereas $c(N[v_{i}])$ contains at least one color of
this set.  Finally, $c(N[v_{i}])\neq c(N[v_{j}])$ for $i,j\leq k-2$
because there is a vertex in $S_2$ that separates them and its color
is used only once. Hence, for each edge $uv$ of $G$ such that
$N[u]\neq N[v]$, we have $c(N[u])\neq c(N[v])$.

\medskip
{\bf Case 2.} For each vertex $x$ of $S\setminus S_1$, either $x$ has
degree at most $k-2$ or $x$ has degree $k-1$ and each vertex of $N(x)$
has a neighbor in $S_1$. We define a coloring $c$ as follows: vertices
of $K$ are assigned colors $1,\ldots,k$, and vertices of $S_1$ are
assigned (pairwise distinct) colors within $k+1,\ldots,2k-1$.  For any
vertex $x$ in $S\setminus S_1$, take a vertex $v_i$ in $K \setminus
N(x)$ (such a vertex exists by the maximality of $K$) and set
$c(x)=c(v_i)$.  We claim that $c$ is a lid-coloring of $G$. It is
clear that $c$ is a proper coloring of $G$. Let $uv$ be an edge of $G$
with $N[u]\neq N[v]$.  If $u,v \in K$, then without loss of generality
there is a vertex $w$ of $S_1$ such that, $w\in N[u]$ and $w \notin
N[v]$. Then, $c(w)\in c(N[u])$ and $c(w)\notin c(N[v])$.  Otherwise,
without loss of generality, $u\in K$ and $v \in S$.  If $v\in S_1$, 
then $S_1$ does not contain the whole set $c(K)$ and so $c(N[u])\neq
c(N[v])$.  Otherwise, $v\notin S_1$. If the degree of $v$ is $k-1$, 
then $u$ has a neighbor $w$ in $S_1$ and $c(w)\in c(N[u])$,
$c(w)\notin c(N[v])$. If the degree of $v$ is at most $k-2$, then there is
a color $1 \le i\le k$ such that $i\in c(N[u])$ and $i\notin c(N[v])$.
In all cases, $c(N[u])\neq c(N[v])$. Hence, $c$ is a lid-coloring of
$G$.
\end{proof}

Observe that this bound is sharp: the graph obtained from a $k$-clique
by adding a pendant vertex to each of the vertices of the clique is a
split graph and requires $2k-1$ colors in any lid-coloring.

\section{Cographs \label{sec:cographs}}

A \emph{cograph} is a graph that does not contain the path $P_4$ on 4
vertices as an induced subgraph. Cographs are a subclass of
permutation graphs, and so they are perfect (however, they are not
necessarily chordal). It is well-known that the class of cographs is
closed under disjoint union and complementation~\cite{BLS99}. Let $G
\cup H$ denote the disjoint union of $G$ and $H$, and let $G +H$
denote the complete join of $G$ and $H$, i.e. the graph obtained from $G
\cup H$ by adding all possible edges between a vertex from $G$ and a
vertex from $H$. A consequence of the previously mentioned facts is
that any cograph $G$ is of one of the three following types:
\begin{enumerate}
\item[(S)] $G$ is a single vertex.
\item[(U)] $G=\bigcup_{i=1}^k G_i$ with $k \ge 2$ and every $G_i$ is
  a cograph of type S or J.
\item[(J)]$G=\sum_{i=1}^k G_i$ with $k \ge 2$ and every $G_i$ is a
  cograph of type S or U.
\end{enumerate} 

We will use this property to prove the following theorem:

\begin{theorem}\label{cograph}
  If $G$ is a cograph, then $\chi_{\text{lid}}(G)\leq 2\omega(G)-1$.
\end{theorem}

\begin{proof}
A {\em universal} vertex of $G$ is a vertex adjacent to all the vertices of $G$.
Observe that if a cograph $G$ has a universal vertex, then $G$ must be of
type S or J. 
Let $\widetilde{\chi}_{\text{lid}}(G)$ be the least integer $k$ such that $G$
has a lid-coloring $c$ with colors $1,\ldots,k$ such that for any
vertex $v$ that is not universal, $c(N[v]) \ne \{1,\ldots,k\}$ (in
other words, if a vertex sees all the colors, then it is
universal). Such a coloring is called a \emph{strong lid-coloring} of
$G$. We will prove the following result by induction:

\medskip
  \emph{Claim. For any cograph $G$, $\chi_{\text{lid}}(G)\le
    2\omega(G)-1$ and $\widetilde{\chi}_{\text{lid}}(G)\le
    2\omega(G)$.}

\medskip
If $G$ is a single vertex, then it is universal and therefore
$\widetilde{\chi}_{\text{lid}}(G)=\chi_{\text{lid}}(G)=1=2\times 1-1$ and the
assumption holds.

Assume now that $G$ is of type J. There exist $G_1,\ldots,G_k$, $k\ge
2$, each of type S or U, such that $G=\sum_{i=1}^k G_i$. Let
$G_1,\ldots,G_s$ ($0 \le s \le k$) be of type S and
$G_{s+1},\ldots,G_k$ be of type U. Consider a lid-coloring $c_1$ of
$G_1$ and a strong lid-coloring $c_i$ of $G_i$ for $2\le i\le k$, such
that the sets of colors within $G_i$ and $G_j$, $i\ne j$, are
disjoint.  Then the coloring $c$ of $G$ defined by $c(v)=c_i(v)$ for
any $v \in G_i$ is a lid-coloring of $G$.  To see this, assume two
adjacent vertices $u$ and $v$ such that $N[u]\ne N[v]$ and
$c(N[u])=c(N[v])$. Since every $c_i$ is a lid-coloring of $G_i$ the
vertices $u$ and $v$ must be in different $G_i$'s, say $u\in G_i$ and
$v \in G_j$, $i < j$. But then in order to have $c(N[u])=c(N[v])$, $u$
and $v$ must see all the colors in $c_i$ and $c_j$, respectively.
Since $c_j$ is a strong lid-coloring of $G_j$, $v$ is universal in
$G_j$.  This means that $G_j$ (and therefore $G_i$) is of type
S. Hence, $u$ and $v$ are universal in $G$. This contradicts the fact
that $N[u]\neq N[v]$. As a consequence $c$ is a lid-coloring of $G$.

If $c_1$ is a strong lid-coloring of $G_1$, then $c$ is a strong
lid-coloring of $G$: take a vertex $v \in G_i$ that sees all the
colors in $c$.  Then it also sees all the colors in $c_i$, so it is
universal in $G_i$ and $G$.

So we have $\chi_{\text{lid}}(G)\le
\chi_{\text{lid}}(G_1)+\sum_{i=2}^k\widetilde{\chi}_{\text{lid}}(G_i)$ and
$\widetilde{\chi}_{\text{lid}}(G)\le \sum_{i=1}^k
\widetilde{\chi}_{\text{lid}}(G_i)$.  Since $\omega(G)=\sum_{i=1}^k
\omega(G_i)$ we have by induction:
$$\chi_{\text{lid}}(G)\le 2\omega(G_1)-1+\sum_{i=2}^k 2 \omega(G_i) =
2\times \sum_{i=1}^k \omega(G_i)-1 = 2 \omega(G)-1$$ and
$$\widetilde{\chi}_{\text{lid}}(G)\le \sum_{i=1}^k 2 \omega(G_i)= 2 \omega(G).$$

\medskip
Assume now that $G$ is of type U. There exist $G_1,\ldots,G_k$, $k\ge
2$, each of type S or J, such that $G=\bigcup_{i=1}^k G_i$.  Consider
a lid-coloring $c_i$ of $G_i$ with colors $1,\ldots,\chi_{\text{lid}}(G_i)$.
Without loss of generality we have $\chi_{\text{lid}}(G_1) = \max_{i=1}^k
\chi_{\text{lid}}(G_i)$.  The coloring $c$ of $G$ defined by $c(v)=c_i(v)$
for any $v\in G_i$ is clearly a lid-coloring of $G$, and so
$\chi_{\text{lid}}(G)=\max_{i=1}^k \chi_{\text{lid}}(G_i)$.

 To obtain a strong lid-coloring, assign a new color
 $\chi_{\text{lid}}(G_1)+1$ to all the vertices colored $1$ in $G_1$, and
 color all the other vertices of $G$ as they were colored in $c$. The
 obtained coloring $c'$ is still a lid-coloring of $G$.  Since no
 vertex $u$ satisfies $c(N[u])=\{1, \ldots, \chi_{\text{lid}}(G_1)+1\}$ (the
 vertices in $G_1$ miss the color 1, while the others miss the color
 $\chi_{\text{lid}}(G_1)+1$), $c'$ is also a strong lid-coloring of
 $G$. Therefore $\widetilde{\chi}_{\text{lid}}(G)\le \max_{i=1}^k
 \chi_{\text{lid}}(G_i)+1$.  Since $\omega(G)=\max_{i=1}^k \omega(G_i)$ we
 have by induction
 $$\chi_{\text{lid}}(G)\le
   \max_{i=1}^k (2\omega(G_i)-1) = 2 \omega(G)-1$$
and  $$\widetilde{\chi}_{\text{lid}}(G)\le
   \max_{i=1}^k (2\omega(G_i)-1)+1 = 2 \omega(G).$$
\end{proof}

The bound of Theorem~\ref{cograph} is tight. The following
construction gives an example of cographs of clique number $\omega$
requiring $2\omega-1$ colors in any lid-coloring. For any $k \ge 1$,
take a complete graph with vertex set $v_1, \ldots, v_k$ and for each
$2 \le i \le k$ add a vertex $u_i$ such that
$N(u_i)=\{v_i,v_{i+1},\ldots,v_k\}$. This graph is a cograph with
clique number $k$, the vertices $u_i$ form an independent set $U$, and
every vertex $v_i$ satisfies $N(v_i)\cap U=\{u_2,\ldots,u_i\}$. Let
$c$ be a lid-coloring of this graph, then for any $3 \le i \le k$ the
vertex $u_i$ must be assigned a color distinct from $c(u_2),\ldots,
c(u_{i-1})$ and $c(v_1),\ldots,c(v_k)$ since otherwise we would have
$c(N[v_i])=c(N[v_{i-1}])$. Hence, at least $k+(k-1)=2k-1$ distinct
colors are required.

\medskip
As mentionned in Section~\ref{sec:ktree}, for fixed $t$, the fact that a graph admits a lid-coloring with at most $t$ colors can be expressed in monadic second-order logic.
It is well known that the class of cographs is exactly the class of graphs with clique-width at most two. 
It follows from~\cite{CMR00} and Theorem~\ref{cograph} that, for a fixed $k$, the lid-chromatic number of a
cograph of clique number at most $k$ can be computed in linear time.

\medskip
Given the results in Sections~\ref{sec:bip} to \ref{sec:cographs}, it
seems natural to conjecture that every perfect graph $G$ has
lid-chromatic number at most $2 \chi(G)$. This is not true, however,
as the following example shows. Take three stable sets $S_1,S_2,S_3$,
each of size $k$ ($k\geq 2$), add all possible edges between $S_1$ and
$S_2$, add a perfect matching between $S_1$ and $S_3$, and add the
complement of a perfect matching between $S_2$ and $S_3$. The obtained
graph $G_k$ is perfect: since the subgraph of $G_k$ induced by $S_1$
and $S_2$ is a complete bipartite graph, an induced subgraph of $G_k$
is bipartite if and only if it does not have a triangle, and is
3-colorable otherwise.

Consider a lid-coloring $c$ of $G_k$, and a vertex $x_2$ of $S_2$. Let
$x_3$ be the only vertex of $S_3$ that is not adjacent to $x_2$, and
$x_1$ be the unique neighbor of $x_3$ in $S_1$. Observe that
$N[x_1]=N[x_3] \cup \{x_2\}$. Since $c(N[x_1]) \ne c(N[x_3])$, the
color of $x_2$ appears only once in $S_2$. Hence, all the vertices of
$S_2$ have distinct colors and it follows that $\chi_{\text{lid}}(G_k)\ge
k+2$, whereas $\chi(G_k)=\omega(G_k)=3$.

\section{Graphs with bounded maximum degree \label{sec:bmaxd}}

\begin{proposition}
If a graph $G$ has maximum degree $\Delta$, then $\chi_{\text{lid}}(G) \le
\Delta^3-\Delta^2+\Delta+1$.
\end{proposition}

\begin{proof}
Let $c$ be a coloring of $G$ so that vertices at distance at most
three in $G$ have distinct colors. Since every vertex has at most
$\Delta^3-\Delta^2+\Delta$ vertices at distance at most three, such a
coloring using at most $\Delta^3-\Delta^2+\Delta+1$ colors exists. Let
$uv$ be an edge of $G$. Let $N_u$ be the set of neighbors of $u$ not
in $N[v]$ and $N_v$ be the set of neighbors of $v$ not in $N[u]$.
Using that vertices at distance at most two in $G$ have distinct
colors, we obtain that all the elements of $N_u$ (resp. $N_v$) have
distinct colors. Since vertices at distance at most three have
distinct colors, the sets of colors of $N_u$ and $N_v$ are disjoint. If
$N[u] \ne N[v]$, then $N_u \cup N_v \ne \varnothing$, and $c(N[u])\neq
c(N[v])$ by the previous remark.
\end{proof}

We believe that this result is not optimal, and that the bound should
rather be quadratic in $\Delta$:

 \begin{question}\label{delta2}
Is it true that for any graph $G$ with maximum degree $\Delta$, we
have $\chi_{\text{lid}}(G) =O( \Delta^2)$?
\end{question}

If true, then this result would be best possible. Take a projective
plane $P$ of order $n$, for some prime power $n$. Let $G_{n+1}$ be the
graph obtained from the complete graph on $n+1$ vertices by adding,
for every vertex $v$ of the clique, a vertex $v'$ adjacent only to
$v$. Note that in any lid-coloring of $G_{n+1}$, all vertices $v'$
must receive distinct colors. For any line $l$ of the projective plane
$P$, consider a copy $G_{n+1}^l$ of $G_{n+1}$ in which the new
vertices $v'$ are indexed by the $n+1$ points of $l$. For any point
$p$ of $P$, identify the $n+1$ vertices indexed $p$ in the graphs
$G_{n+1}^l$, where $p \in l$, into a single vertex $p^*$. The
resulting graph $H_{n+1}$ is $(n+1)$-regular and has $(n^2+n+1)(n+2)$
vertices.  By construction, all the vertices $p^*$, $p \in P$, have
distinct colors in any lid-coloring. Hence, at least
$n^2+n+1=\Delta^2-\Delta+1$ colors are required in any lid-coloring of
this $\Delta$-regular graph. The 3-regular graph $H_3$ with
$\chi_{\text{lid}}(H_3) \ge 7$ is depicted in
Figure~\ref{fig:clique3}.

\begin{figure}[htbp]
  \begin{center}
    \includegraphics[scale=0.6]{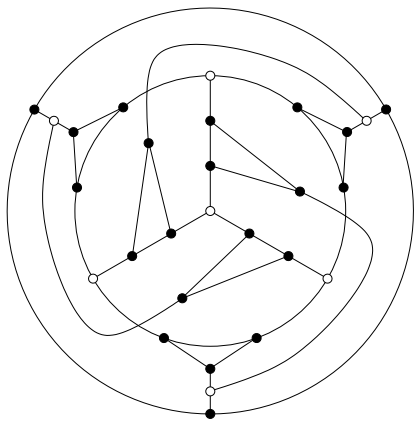}
    \caption{In any lid-coloring of the 3-regular graph $H_3$, the
      seven white vertices must receive pairwise distinct colors. 
      \label{fig:clique3}}
  \end{center}
\end{figure}

\medskip
We saw that the lid-chromatic number cannot be upper-bounded by the
chromatic number. For a graph $G$, the square of $G$, denoted by
$G^2$, is the graph with the same vertex set as $G$, in which two
vertices are adjacent whenever they are at distance at most two in
$G$. The following question is somehow related to the previous one
(depending on the possible linearity of $f$).

\begin{question}
Does there exist a function $f$ so that for any graph $G$, we have
$\chi_{\text{lid}}(G) \le f(\chi(G^2))$?
\end{question}

\section{Planar and outerplanar graphs \label{sec:planar}}

This section is devoted to graphs embeddable in the plane. A maximal
outerplanar graph is a $2$-tree and so is $6$-lid colorable by
Theorem~\ref{th-ktree}. However, $\chi_{\text{lid}}$ is not monotone under
taking subgraphs and so this result does not extend to all outerplanar
graphs. So we have to use a different strategy to give an upper bound
of the lid-chromatic number on the class of outerplanar graphs.

\begin{theorem}
Every outerplanar graph is $20$-lid-colorable.
\end{theorem}

\begin{proof}
Let $G$ be a connected outerplanar graph, and let $H$ be any maximal
outerplanar graph containing $G$ (that is, $H$ is obtained by adding
edges to $G$). The graph $H$ is 2-connected, and has minimum degree
two. Consider a drawing of $H$ in the plane, such that all the
vertices lie on the outerface, and take the clockwise ordering
$x_1,\ldots,x_n$ of the vertices around the outerface, starting at
some vertex $x_1$ of degree two in $H$ (and thus at most two in
$G$). This ordering has the following properties:

\begin{itemize}
\item For any four integers $i,j,k,\ell \in \{1,\ldots,n\}$ with
  $i<j<k<\ell$, at most one of the pairs $\{x_i,x_k\}$ and $\{x_j,x_\ell\}$
  corresponds to an edge of $G$.
\item Let $x_{i_0}$ be a vertex and $x_{i_1},\ldots,x_{i_k}$ be its
  neighbors in $G$ such that $x_{i_0},x_{i_1},\ldots,x_{i_k}$ appear
  in clockwise order around the outerface of $H$. The previous
  property implies that, for $1\leq j \leq k$, the neighbors of
  $x_{i_j}$ distinct from $x_{i_0}$ appear (in clockwise order around
  the outerface of $H$) between $x_{i_{j-1}}$ and $x_{i_{j+1}}$ (if
  $j\neq k$) and between $x_{i_{k-1}}$ and $x_{i_0}$ (if
  $j=k$). Moreover, two distinct vertices $x_{i_j}$ and $x_{i_\ell}$
  have at most one common neighbor outside $N[x_{i_0}]$. If such a
  common neighbor exists, then we have $|j-\ell|=1$.
\end{itemize}

   For any $i\ge 1$, let
  $L_i=\{x_{i_1},\ldots,x_{i_{k_i}}\}$ be the set of vertices at
  distance $i$ from $x_1$ in $G$, with $i_1<\cdots<i_{k_i}$, and let
  $L_s$ be the last nonempty $L_i$-set. For the sake of clarity, we
  write $x_1^i,\ldots,x_{k_i}^i$ instead of
  $x_{i_1},\ldots,x_{i_{k_i}}$, and we say that two vertices $x_{j}^i$
  and $x_{j+1}^i$ are \emph{consecutive} in $L_i$. Observe the
  following:
\begin{itemize}
\item A vertex in $L_{i+1}$ has at most two neighbors in $L_i$.
\item Two vertices of $L_i$ have at most one common neighbor in $L_{i+1}$.
\item If two vertices of $L_i$ have a common neighbor in $L_{i+1}$,
  they are consecutive in $L_i$.
\item If two vertices of $L_i$ are adjacent, then they are consecutive
  in $L_i$. This implies that the graph induced by $L_i$ is a disjoint
  union of paths.
\end{itemize}

Indeed, if one of the two first facts was not true, there would be a
subdivision of $K_{2,3}$ in $G$. The two last facts are due to the
embedding of $G$ and $H$ and to the previous properties.  From now on, we
forget about $H$ and consider $G$ only (the sole purpose of $H$ was to
give a clean definition of the order $x_1,\ldots,x_n$). With the facts
above, we can notice that in the ordering of $L_{i+1}$, we find first
the neighbors of $x_1^i$, then the neighbors of $x_2^i$, and so on...

We will color the vertices of $G$ with 20 colors partitioned in four
classes of colors $C_0$, $C_1$, $C_2$ and $C_3$ with
$C_j=\{5j,\ldots,5j+4\}$.  Vertices in $L_i$ will be colored with
colors from $C_{i\bmod 4}$, almost as we did for bipartite graphs in
Theorem \ref{th-bip}. We will slightly modify this coloring by using
{\em marked vertices}.

We start by coloring $x_1$ with color 0, and mark the last vertex
$x_{k_1}^1$ of $L_1$. We then apply Algorithm 1.

\begin{algorithm}[h]
\caption{Lid-coloring of outerplanar graphs}
\begin{algorithmic}[1]
\STATE $c(x_1)=0$
\STATE Mark vertex $x^1_{k_1}$
\FOR{$i=1$ \TO $s$}\label{line:li}
  \FOR{$j=1$ \TO $k_i$} \label{line:mark1}
  \STATE Mark, if it exists, the last neighbor of $x^i_j$ in $L_{i+1}$.
  \ENDFOR\label{line:mark2}
  \FOR{$k=0$ \TO $3$}
  \STATE $c_{k}\leftarrow k+ 5\times (i\bmod4)$ 
  \ENDFOR
  \STATE $c_{\varnothing}\leftarrow 4+5\times (i\bmod4)$ 
  \FOR{$j=1$ \TO $k_i$} 
  \STATE $c(v^i_j)=c_{j\bmod 4}$ \label{line:coloring}
  \IF{$v^i_j$ is marked}\label{line:change1}
  \STATE tmp $\leftarrow c_{(j+1)\bmod 4}$
  \STATE $c_{(j+1)\bmod 4} \leftarrow c_{\varnothing}$
  \STATE $c_{\varnothing}\leftarrow c_{(j-1)\bmod 4}$
  \STATE $c_{(j-1)\bmod 4} \leftarrow$ tmp
  \ENDIF\label{line:change2}
 \ENDFOR
\ENDFOR
\RETURN $c$
\end{algorithmic}
\end{algorithm}

\medskip 
Let us describe this algorithm. Sets $L_i$ are colored one after the
other (line \ref{line:li}).  When we color $L_i$, we first mark some
vertices in $L_{i+1}$ (the last neighbors in $L_{i+1}$ of vertices in
$L_i$, see lines \ref{line:mark1} to \ref{line:mark2}).  Then we color
vertices of $L_i$ in the order they appear.  There are four current
colors of $C_{i\bmod 4}$ which are used, $c_0$ to $c_3$ and one
forbidden color $c_{\varnothing}$, that are originally set to $5\times
(i\bmod4)$, $1+5\times (i\bmod4)$, $2+5\times (i\bmod4)$, $3+5\times
(i\bmod4)$, and $4+5\times (i\bmod4)$, respectively. The vertices of
$L_i$ are then colored with the pattern $c_0c_1c_2c_3c_0$...  (line
\ref{line:coloring}), but every time a marked vertex $v_j^i$ is
colored, we perform a cyclic permutation on the values of
$c_{(j+1)\bmod 4}$, $c_{\varnothing}$, and $c_{(j-1)\bmod 4}$ (lines
\ref{line:change1} to \ref{line:change2}). This is done in such a way
that:
\begin{itemize}
\item The coloring is proper.
\item Four consecutive vertices in $L_i$ receive four different
  colors.
\item Two consecutives vertices of $L_{i-1}$ do not have the same set
  of colors in their neighborhood in $L_i$, when these neighborhoods
  differ.
\end{itemize}

Thus, this algorithm provides a proper coloring $c$ of $G$ with $20$
colors such that for any $i$, $c(L_i) \subseteq C_{i \bmod 4}$.

\medskip 
Let us prove that the coloring given by the algorithm is locally
identifying.  Let $uv$ be an edge of $G$ such that $N[u]\neq N[v]$. If
$uv$ is not an edge of a layer $L_i$, then we can assume that $u\in
L_i$ and $v\in L_{i+1}$. If $u\neq x_1$, then there is a neighbor $t$
of $u$ in $L_{i-1}$ and then $c(t)\notin c(N[v])$.  So we may assume
that $u=x_1$. If the vertex $v$ has degree 1, then $u$ has degree $2$
and has an other neighbor, $t$, and $c(t)\notin c(N[v])$.  Otherwise,
the vertex $v$ has degree at least $2$, so there is a neighbor $t\neq
u$ of $v$.  If $t\in L_1$ then there is another neighbor $t'$ of $v$
in $L_2$ (because $N[u]\neq N[v]$).  So we can assume that $t\in L_2$
and then $c(t)\notin c(N[u])$. So in any case, $c(N[u])\neq c(N[v])$.

\medskip 
Assume now that $u,v \in L_i$ for some $i$. Without loss of
generality, we may assume that $u=x^i_j$, $v=x^i_{j+1}$ for some $j$
and that there is a vertex $t$ adjacent to exactly one vertex among
$\{u,v\}$.  If $t\in L_i$, then we are done because four consecutive
vertices have different colors in $L_i$. If $t\in L_{i-1}$, and $t\in
N(u)\setminus N(v)$, then $v$ has at most two neighbors in $L_{i-1}$.
Those neighbors (if any) are just following $t$ in the layer $L_{i-1}$
and so $c(t)\notin c(N[v])$.  Otherwise, $t\in L_{i+1}$, the vertices
$u$ and $v$ are consecutive and have distinct neighborhoods in
$L_{i+1}$, so the sets of colors in their neighborhoods in $L_{i+1}$
are distinct.
\end{proof}

We believe that this bound is far from tight.

\begin{question}
Is it true that every outerplanar graph $G$ satisfies
$\chi_{\text{lid}}(G)\leq 6$?
\end{question}

We now prove that sparse enough planar graphs have low lid-chromatic
number.

\begin{theorem} \label{p+g} If $G$ is a planar graph with girth at
  least 36, then $\chi_{\text{lid}}(G)\le 5$.
\end{theorem}

\begin{proof}
Let us call \emph{nice} a lid-coloring $c$ using at most 5 colors
  such that every vertex $v$ with degree at least 2 satisfies
  $|c(N[v])|= 3$.
We show that every planar graph with girth at least 36 admits a nice
lid-coloring.

\medskip 
Observe first that a cycle of length $n\geq 12$ has a nice
lid-coloring that consists of subpaths of length 4 colored 1234 and
subpaths of length 5 colored 12345 following the clockwise orientation
of $G$ (the number of subpaths of length $5$ is exactly $n \bmod 4$).

\medskip
Suppose now that $G$ is a planar graph with girth at least 36 that does
not admit a nice lid-coloring and with the minimum number of vertices.
Let us first show that $G$ does not contain a vertex of degree at most
1.  The case of a vertex of degree 0 is trivial, so suppose that $G$
contains a vertex $u$ of degree 1 adjacent to another vertex $v$.  By
minimality of $G$, the graph $G'=G\setminus u$ admits a nice
lid-coloring $c$.  We consider three cases according to the degree of
$v$ in $G'$, and in all three cases, we extend $c$ to a nice
lid-coloring of $G$ in order to obtain a contradiction. If $v$ has
degree at least 2 in $G'$, then we assign to $u$ a color in
$c(N[v])\setminus\{c(v)\}$. So $c(N[v])$ is unchanged, and $c(N[u])\ne
c(N[v])$ since $|c(N[u])|=2$ and $|c(N[v])|=3$. We thus have a nice
lid-coloring of $G$. If $v$ has degree 1 in $G'$, then $v$ is adjacent
to another vertex $w$ in $G'$ and we assign to $u$ a color that does
not belong to $c(N[w])$.  Such a color exists since $|c(N[w])|\leq 3$ and
the obtained coloring of $G$ is nice: $|c(N[v])|=3$ and $c(N[v])\ne
c(N[w])$ since $c(u)\in c(N[v])$ but $c(u)\not\in c(N[w])$. If $v$ has
degree 0 in $G'$, then $N[u]=N[v]$ in $G$, so $u$ and $v$ need not to be identified.

\medskip 
It follows that $G$ has minimum degree at least 2 and $G$ is not a
cycle.  It is well-known that if the girth of a planar graph is at
least $5k+1$, then it contains either a vertex of degree at most 1, or
a path consisting of $k$ consecutive vertices of degree 2. The graph
$G$ thus contains a path of seven vertices of degree 2. So we can
assume that $G$ contains a path $P=x_1x_2\ldots x_9$ such that
$d(x_1)\geq 3$ ($G$ is not a cycle), $d(x_i)=2$ for $2\leq i \leq 8$,
and $d(x_9)\geq 2$.  By minimality of $G$, the graph
$G'=G\setminus\{x_2,x_3,\ldots,x_8\}$ admits a nice lid-coloring
$c$. Without loss of generality, assume that $c(x_1)=1$ and
$c(N[x_1])=\acc{1,2,3}$, since the degree of $x_1$ is at least $2$ in
$G'$.  We denote $a=c(x_9)$. If the degree of $x_9$ in $G'$ is at
least 2, then we denote $\acc{b_1,b_2}=c(N(x_9))$. If the degree of
$x_9$ in $G'$ is 1, then $x_9$ is adjacent to a vertex $x_{10}$ and we
denote $b_1=c(x_{10})$ and $b_2$ is any element of
$\acc{1,2,3,4,5}\setminus c(N[x_{10}])$.

 The following table gives the colors of $x_2,x_3,\ldots,x_{8}$ for
 all the possible values of $\paren{a;b_1,b_2}$. Note that $c(x_2)\in
 \{2,3\}$, $c(x_3)\notin \{2,3\}$, $c(x_6)\neq a$, $c(x_7)\notin
 \{a,b_1,b_2\}$, $c(x_8)=b_2$, and four consecutive vertices have
 different colors. This implies that the coloring $c$ can be extended
 to a nice lid-coloring of $G$, a contradiction.

\begin{center}
\begin{tabular}{|c|c|c|c|c|}
\hline 2431243 (1;2,3) & 2431543 (2;1,3) & 2431542 (3;1,2) & 2534152
(4;1,2) & 2435142 (5;1,2)\\ \hline 2431254 (1;2,4) & 2541354 (2;1,4) &
2431254 (3;1,4) & 2431253 (4;1,3) & 2431243 (5;1,3)\\ \hline 2431245
(1;2,5) & 2451345 (2;1,5) & 2431245 (3;1,5) & 2451235 (4;1,5) &
2435124 (5;1,4)\\ \hline 2431254 (1;3,4) & 3512354 (2;3,4) & 2431254
(3;2,4) & 2431253 (4;2,3) &2431243 (5;2,3)\\ \hline 2431245 (1;3,5) &
3412345 (2;3,5) & 2431245 (3;2,5) & 2451235 (4;2,5) & 2435214
(5;2,4)\\ \hline 2531425 (1;4,5) & 3521435 (2;4,5) & 2531425 (3;4,5) &
2534125 (4;3,5) & 2435124 (5;3,4)\\ \hline
\end{tabular}
\end{center}
\end{proof}

We conjecture that planar graphs have bounded lid-chromatic number.

\section{Connectivity and lid-coloring \label{conn}}

Most of the proofs we gave in this article heavily depend on the
structure of the classes of graphs we were considering. We now give a
slightly more general tool, allowing us to extend results on the
$2$-connected components of a graph to the whole graph:

\begin{theorem}\label{2connto1conn}
Let $k$ be an integer and $G$ be a graph such that every $2$-connected
component of $G$ is $k$-lid-colorable.  Let $H$ be the subgraph of $G$
induced by the cut-vertices of $G$.  Then $\chi_{\text{lid}}(G)\leq
k+\chi(H)$.
\end{theorem}

\begin{proof}
In this proof, we will consider two different colorings of the
vertices: the lid-coloring of the vertices of $G$ and the proper
coloring of the graph $H$ induced by the cut-vertices. To avoid confusion,
we call {\em type} the color of a cut-vertex in the second
coloring. We prove the following stronger result:

\medskip
{\it Claim: If $t$ is a proper coloring of $H$ with colors
  $t_1,\ldots,t_h$, then $G$ admits a $(k+h)$-lid-coloring $c$ such that
  for each maximal $2$-connected component $C$ of $G$, $(*)$  there are
  $h$ colors not appearing in $c(C)$, say $c_1^C,\ldots,c_{h}^C$, such
  that for every cut-vertex $v$ of $G$ lying in $C$, if $t(v)=t_i$,
  then $c(N(v))$ contains $c^C_i$ but none of the $c^C_j$, $j \ne
  i$.}

\medskip
We prove the claim by induction on the number of cut-vertices of
$G$. We may assume that $G$ has a cut-vertex, otherwise the property
is trivially true.

Let $u$ be a cut-vertex of $G$ and let $C_1,\ldots,C_s$ be the
connected components of $G-u$. We can choose $u$ so that at most one
of the $C_i$'s, say $C_1$, contains the remaining cut-vertices. For
$1\leq i \leq s$, let $G_i$ be the graph induced by the set of
vertices $C_i\cup\{u\}$. Let $C$ be the maximal $2$-connected
component of $G_1$ containing $u$. Observe that the vertex $u$ is not
a cut-vertex in $G_1$. By the induction, $G_1$ has a
$(k+h)$-lid-coloring $c$ such that, without loss of generality,
$c(C)\subseteq \{1,\ldots,k\}$ and every cut-vertex $v$ of $C$ with
$t(v)=t_i$ has a neighbor colored $k+i$, but no neighbor colored
$k+j$, $1\leq j \leq h$, $j \ne i$.  We can also assume that
$t(u)=t_1$ and $1\in c(N(u))$ (thus $c(u)\neq 1$).

We now extend the coloring $c$ to $G$ by lid-coloring each component
$G_2,...G_s$ with colors $2,3,\ldots,k+1$ such that $k+1\in c(N(u))$
(these components share the vertex $u$ but we can assume that $u$
always has the same color in all the lid-colorings of
$G_2$,...,$G_s$). Let us prove that the coloring obtained is a
lid-coloring of $G$ satisfying $(*)$. In order to prove that $c$ is a
lid-coloring, by the induction one just needs to check that $u$ has no
neighbor $v$ with $c(N[v])=c(N[u])$. For the sake of contradiction,
suppose that such a vertex $v$ exists. Since $1\in c(N[u])$, $v$ has
to lie in $C$. If $v$ is a cut-vertex of $G_1$, then $t(v)\ne t_1$
($t$ is a proper coloring of $H$) and by the induction, $k+1 \not\in
c(N[v])$. If $v$ is not a cut-vertex of $G_1$, then all its neighbors
lie in $C$ and again, $k+1 \not\in c(N[v])$. Since $k+1 \in c(N[u])$,
we obtain a contradiction.

It remains to prove that $(*)$ holds for every maximal 2-connected
component of $G$. It clearly does for $G_2$,...,$G_s$, since $u$
is the only cut-vertex of $G$ that they contain and $1 \in
c(N[u])\subseteq \{1,\ldots,k+1\}$, while none of these components
contains color 1 or color $k+i$ with $2\le i \le h$. The component $C$
also satisfies $(*)$, since $u$ has a neighbor colored $k+1$ and no
neighbor colored $k+i$ with $2\le i \le h$. By the induction, Property
$(*)$ trivially holds for the remaining maximal 2-connected components
of $G$. This completes the proof of the claim.
\end{proof}

Among other things, this result can be used to prove that outerplanar graphs
without triangles can be 8-lid-colored. We omit the details; we
suspect that Theorem~\ref{2connto1conn} can be used to prove results
on much wider classes of graphs.

\paragraph{Remark.}During the review of the paper,
Question~\ref{delta2} has been answered positively,
see~\cite{FHLPP12}. 

\paragraph{Acknowledgements.} We would like to acknowledge E.~Duch\^ene about early discussions on the topic of identifying coloring, which inspired this work.
We also would like to thank the referees for careful reading and helpful remarks.

\end{document}